\documentclass{article}

\usepackage{notes}

\title{Classical and Quantum Algorithms for \\
  Testing Equivalence of Group Extensions}
\author{Kevin C. Zatloukal \\ University of Washington}

\newcommand{\poly}{{\rm poly}}

\begin{document}

\maketitle

\begin{abstract}
While efficient algorithms are known for solving many important problems related
to groups, no efficient algorithm is known for determining whether two arbitrary
groups are isomorphic. The particular case of 2-nilpotent groups, a special type
of central extension, is widely believed to contain the essential hard cases.
However, looking specifically at central extensions, the natural formulation of
being ``the same'' is not isomorphism but rather ``equivalence,'' which requires
an isomorphism to preserves the structure of the extension. In this paper, we
show that equivalence of central extensions can be computed efficiently on a
classical computer when the groups are small enough to be given by their
multiplication tables. However, in the model of black box groups, which allows
the groups to be much larger, we show that equivalence can be computed
efficiently on a quantum computer but not a classical one (under common
complexity assumptions). Our quantum algorithm demonstrates a new application of
the hidden subgroup problem for general abelian groups.
\end{abstract}

\section{Introduction}

Finding an efficient algorithm for group isomorphism is one of the most notable
open problems in computational group theory. While the problem is easily solved
for abelian groups, the problem remains unsolved even for some very simple
generalizations to non-abelian groups. In particular, the 2-nilpotent groups,
which are central extensions of an abelian group by another abelian group, are
widely believed to contain the essential hard cases (see e.g.
\cite{TwoNilGroupIso}). Hence, the computational issues surrounding this type of
group extension merit further study.

While isomorphism is the natural notion of what it means to be the same group,
the natural notion of being the same extension is slightly different. Indeed,
the theory of group extensions\footnote{See the chapter in \cite{IntroGroups}
for a nice introduction to the theory of group extensions.}, whose study began
near the start of the 20th century, defines two extensions to be the same or
``equivalent'' if there exists an isomorphism that preserves the structure of
the extension. (We will define this precisely in the next section.)

Thus, it is interesting to consider whether there exists an efficient algorithm
for testing equivalence of those extensions for which isomorphism appears
difficult. In this paper, we will see that there is indeed an efficient
algorithm.

Group isomorphism has drawn particular interest from the quantum computing
community due to its placement in the hierarchy of complexity classes. In
particular, due to the work of \cite{GroupIsoNPcoNP}, we know that the
isomorphism problem for solvable groups is almost in the class ${\rm NP} \cap
{\rm coNP}$. This is the class that includes factoring and other problems for
which quantum computers appear to give super-polynomial speedups. Hence, there
is strong interest in determining whether the same is true of solvable group
isomorphism. To date, however, no such quantum speedup is known even for the
smaller class of 2-nilpotent groups.

Given the relationship between the conjectured hard cases of group isomorphism
(2-nilpotent groups) and the problem of extension equivalence, it is natural to
wonder whether the latter problem also could lead to a super-polynomial speedup
of quantum algorithms over classical ones. As noted above, there is an efficient
classical algorithm for testing equivalence. However, its efficiency depends on
the fact that the given groups are small, in particular, small enough to write
down their complete multiplication tables.

The usual setting for the group isomorphism problem has the input groups given
by their multiplication tables. If one cannot solve the problem in this model,
then other models are out of the question. However, it would be both interesting
and useful to be able to test equivalence of larger groups, for which this model
is inappropriate. In particular, for groups of matrices over finite fields
(which includes, for example, simple groups of Lie type), individual matrices
are small enough to multiply and invert efficiently, but writing out a
multiplication table between all matrices in the group would often be
infeasible. Yet, computational group theorists would still like to answer
questions about such groups.

Matrix groups are often studied in the ``black box group'' model. (Indeed, this
was the original motivation for the model.) Hence, it is natural for us to
consider whether there exists an efficient algorithms for testing equivalence of
group extensions in this model.

One case we will consider is extending a group given by a multiplication table
by a black box group.  In practical terms, this means extensions of a small
group by a large one. Such extensions can already introduce substantial
complexity. For example, the dihedral group $D_{2N}$ is an extension of the tiny
group $\Integer_2$ by a potentially large cyclic group $\Integer_N$. Considering
that the hidden subgroup problem can be solved in quantum polynomial time for
$\Integer_N$ but not (currently) for $D_{2N}$, we can see that extensions of
even constant-sized groups can introduce substantial computational difficulty.

In this paper, we show that there is an efficient quantum algorithm for testing
equivalence of extensions of a small group by large abelian group or extensions
of one large abelian group by another large abelian group. Furthermore, we will
show that the existence of an efficient classical algorithm for either of these
cases would break an existing cryptosystem.\footnote{Note that this cryptosystem
depends on the hardness of factoring, so it is already known that quantum
computers could break it. What was not known is the relationship of this to
testing equivalence of extensions.} Hence, under the hardness assumption of that
cryptosystem, no efficient classical algorithm exists.

The quantum algorithm we present depends crucially on the ability to solve the
hidden subgroup problem (HSP) for arbitrary abelian groups. (This is the
essential quantum subroutine in our algorithm.) Interestingly, while some other
problems in computational group theory that can be solved efficiently on a
quantum computer can also be solved classically assuming the existence of
oracles for factoring and/or discrete logarithm, our construction does not
easily translate to that setting because there is no apparent way to solve
abelian HSP classically, even with the help of such oracles. Hence, our work
demonstrates a new and interesting application of efficient quantum algorithms
for abelian HSP.

\subparagraph*{Related Work}

While we are aware of no prior work on the complexity of
determining extension equivalence in these models, our motivation for this
problem comes from the status of the group isomorphism problem for simple group
extensions, and there, it is known that isomorphism can be determined
efficiently on a quantum computer in certain special cases
\cite{SpecialGroupIso}. Interestingly, the groups to which this result applies
have trivial equivalence classes\footnote{This follows from the fact that the
second cohomology groups (defined below) are trivial for semi-direct products.},
so the extension equivalence problem is trivial for such groups. (The answer is
always ``yes''.) The fact that the one class of nonabelian solvable groups for
which we have made progress on group isomorphism is one for which equivalence is
trivial suggests that studying the extension equivalence problem may teach us
something about the hard cases of group isomorphism.

\section{Background}

\subsection{Computational Group Theory}

The study of algorithms and complexity for problems in group theory is called
\emph{computational group theory}. In order to discuss these issues, we must
first specify how the group will be given as input. Multiple approaches have
been defined (see \cite{PermGroupAlg} for a nice review).  We will need to use
three of these in our later discussion.

The first approach is to describe a group $G$ by its multiplication table
(sometimes called the ``Cayley table''). Multiplication of group elements can be
performed by table lookup, inverses can be computed by scanning one row of the
table, and so on. This is perhaps the most natural model. However, in order to
use this approach, the group must be small enough that it is reasonable to write
down a $\abs{G} \times \abs{G}$ table. This turns out to be too limiting for
many computations that practitioners want to perform.

Another approach is the ``black box group'' model of Babai and Szemer\'edi"
\cite{BlackBoxGroups}. In this model, group elements are identified by opaque
strings (which need not be unique) and an oracle is provided that can perform
the following group operations:
\begin{enumerate}
\item Given $g, h \in G$, compute $gh$.
\item Given $g \in G$, compute $g^{-1}$.
\item Given $g \in G$, determine whether $g = e$, the group identity
\footnote{This also allows us to determine whether $g = h$ since this is
equivalent to checking $gh^{-1} = e$.}.
\end{enumerate}
Finally, we have to specify how the algorithm obtains the strings for some group
elements in the first place. It is usual to assume that the input to algorithm
will be a list of generators of the group (i.e., a list of strings identifying
the generators).

While the black box model is restricted in terms of how it can work with the
group, it is even more restricted in terms of what is considered efficient.
Since a multiplication table has size $\tilde{O}(\abs{G}^2)$,\footnote{As is
usual, $\tilde{O}(.)$ is the same as $O(.)$ but with suppressed terms that are
logarithmically smaller than those included.} any running time of $\poly
(\abs{G})$ is efficient in the first model.  On the other hand, a
non-redundant\footnote{This simply means that no proper subset of the generators
still generates the group.} list of generators only has length $O(\log
\abs{G})$,\footnote{This follows from the fact that each additional generator
increase the size of the generated group by a factor equal to the index of the
old group in the new one, and this index (an integer), since it is not 1, must
be at least 2.} so the input has size $O(\log^2 \abs{G})$.  Hence, an
algorithm is efficient in the second model only if it has running time
$\poly(\log \abs{G})$, which is exponentially faster.

It should not be surprising then to find a large difference between which
problems can be solved in the two models. In the first model, almost every
natural group problem can be solved efficiently, the notable exception being the
group isomorphism problem. In the second model, on the other hand, very few
problems can be solved, at least classically. The main example of a problem that
can be solved in this model is computing a derived series for a solvable group
(that is, generators for each group in the series) or a central series for a
nilpotent group.

Interestingly, it is known that quantum algorithms can do more in the black box
model. In particular, for abelian or even solvable groups
\cite{QuantumAlgSolvGroups}, a large number of problems can be solved, the most
important example being computing the size of the group, $\abs{G}$. We will show
later on that the extension equivalence problem is another example.

The other approaches for specifying groups use representations of particular
types. The most common of these, the third model we will need below, is to use a
permutation representation. Specifically, we assume that the group is explicitly
a subgroup of the symmetric group, $G \le S_n$. The input is a set of generators
of $G$, each of which is a permutation of the set $[n] \triangleq \{1, \dots,
n\}$.

As in the black box model, $G$ can be specified by at most $O(\log \abs{G})$
generators. Each generator in the input has size $O(n \log n)$, so the input as
a whole will have size $O(n \log n \log \abs{G})$. For an algorithm to be
efficient then, its running time must be polynomial both in $n$ and $\log
\abs{G}$.  Furthermore, for this model to be useful, the size $n$ of the set,
called the ``degree'' of the representation, must be small. The fact that many
groups have small-degree representations is one factor leading to the great
success of this third approach. The other factor leading to its success is that
many problems can be solved efficiently in this model. In fact, nearly all of
the problems that are solvable with multiplication tables are efficiently
solvable here as well. (See \cite{PermGroupAlg} for a long list of these
problems.)

\subsection{Group Extensions}

A group $E$ is said to be an \emph{extension} of $G$ by $A$ if $A \lhd E$ and $E
/ A \cong G$. This is called a \emph{central extension} if $A \le Z(E)$. In
particular, this means that $A$ is abelian.

Central extensions are in some ways similar to semidirect products in that the
elements can be thought of as pairs $(a, x) \in A \times G$ with a strange
multiplication. Whereas multiplication in a semidirect product depends on a
group homomorphism $G \rightarrow \Aut A$, multiplication in a central extension
depends on a function $f : G \times G \rightarrow A$, where we have $(a,x)(b,y)
= (abf(x,y), xy)$. The function $f$ is called a ``factor set.'' We will describe
some of its properties below. In particular, we will show how to find $f$ for a
given extension $E$.

Central extensions are in some sense the other natural way to combine groups,
aside from semidirect products. In particular, any group extension of $G$ by
$A$, where $A$ is abelian but not necessarily central, is essentially a
combination of a semidirect product and a central extension.\footnote{Any
extension is identified, up to isomorphism, by a homomorphism from $G$ to $\Aut
A$ (the semi-direct product part) and a factor set (the central extension part).
See \cite{IntroGroups} for details.} Hence, these two types represent the two
extremes of extensions of abelian groups.

Finally, we can define the problem we are trying to solve. Two extensions, $E_1$
and $E_2$, of $G$ by $A$ are said to be \emph{equivalent} if there exists an
isomorphism $\gamma : E_1 \rightarrow E_2$ such that $\gamma$ is the identity on
$A$, $\gamma |_A = {\rm id}$, and gives rise to the identity on $G$, that is,
$\pi_2 \circ \gamma = \pi_1$, where $\pi_i : E_i \rightarrow G$ is the canonical
projection. This is the natural sense in which two extensions should be
considered ``the same''.

On the other hand, it is possible for $E_1$ and $E_2$ to be isomorphic even if
they are not equivalent extensions. (Indeed, this is not even a simple matter of
dealing with isomorphisms of $A$ and $G$: it is apparently possible for
extensions of non-isomorphic groups to be isomorphic.) For this reason,
equivalence is a more natural question to consider when looking specifically at
group extensions: an equivalence is an isomorphism that respects the structure
of the group extension.

\subsection{Low Degree Group Cohomology}

Cohomology groups are often defined in an abstract manner (via $\rm Ext$
functors, projective resolutions, etc.). However, in the case of group
cohomology, the low degree cohomology groups also have concrete definitions that
are equivalent but more useful for us.\footnote{Historically, these were
developed in the opposite order.  The concrete definitions came first and the
abstract later.} (See \cite{IntroGroups} for a more detailed discussion.)

In this section, we will consider cohomology only of central extensions.
Cohomology can be defined more generally, but this simpler case is all that we
will need in later sections.

The key group for us is the second cohomology group, $H^2(G, A)$. In order
to define this, however, we first need to define cocycles and coboundaries.

The 1-cocycles, $Z^1(G, A)$, are functions $f : G \rightarrow A$ that satisfy
the identity $f(x) + f(y) - f(xy) = 0$, for all $x, y \in G$.  These are simply
group homomorphisms. (Note that we are using additive notation since $A$ is
abelian.) The 2-cocycles, $Z^2(G, A)$, are functions $f : G \times G \rightarrow
A$ that satisfy the (admittedly odd-looking) identity $f(y,z) - f(xy, z) +
f(x,yz) - f(x,y) = 0$, for all $x, y, z \in G$, and have $f(x,e) = f(y,e) = e$,
for all $x, y \in G$.\footnote{Sometimes cocycles are defined only by the first
condition. Then those that satisfy the second are called ``normalized''. We will
assume throughout this paper that all cocycles, coboundaries, and cochains are
properly normalized.} These are precisely the factor sets mentioned earlier.

The 2-coboundaries, $B^2(G, A)$, are functions $G \times G \rightarrow A$ that
arise by taking a function $s : G \rightarrow A$ only satisfying $s(e) = e$
(called a 1-cochain) by defining $\partial s \in B^2(G, A)$ by $\partial s(x,y)
= s(x) + s(y) - s(xy)$. Note that, since $s$ is not necessarily a homomorphism,
we need not have $\partial s \not\equiv 0$. It is not hard to show that any
function defined in such manner is also a 2-cocycle.  In other words, we have
$B^2(G, A) \le Z^2(G, A)$.  Furthermore, the function $\partial$ is in fact a
(surjective) homomorphism $C^1(G, A) \rightarrow B^2(G, A)$, where $C^1(G, A)$
denotes the space of all cochains.

\begin{figure}[t]
\begin{center}
\begin{tabular}{|l|l|} \hline
1-cochains & $C^1(G,A) = \set{s : G \rightarrow A}{s(e) = e}$ \\ \hline
2-cochains & $C^2(G,A) = \set{f : G \times G \rightarrow A}{f
    \text{ normalized}}$ \\ \hline
cocycles & $Z^2(G,A) = \set{f : G \times G \rightarrow A}{f
    \text{ normalized, cocycle condition}} \subset C^2(G,A)$ \\ \hline
$\partial : C^1 \rightarrow C^2$ & homomorphism taking $s \in C^1(G,A)$ to
    $\partial s \in Z^2(G,A)$ \\ \hline
coboundaries & $B^2(G,A) = \Image \partial \subset Z^2(G,A)$ \\ \hline
\end{tabular}
\end{center}
\caption{The main objects in group cohomology.}
\FigureName{group-cohomology}
\end{figure}

These definitions are summarized in \Figure{group-cohomology}.

The sets $Z^2(G, A)$ and $B^2(G, A)$ are themselves groups with the group
operation performed pointwise (i.e., $(f + g)(x, y) = f(x,y)+ g(x,y)$). In fact,
they are abelian groups since $A$ is abelian. Hence, $B^2(G, A)$ is a normal
subgroup of $Z^2(G, A)$, so we can consider the quotient group $H^2(G, A)
\triangleq Z^2(G, A) / B^2(G, A)$. This is the \emph{second cohomology group}.

The most important fact for us is the relationship between $H^2(G, A)$ and group
extensions.

\begin{lemma}
\LemmaName{Correspondence}
Elements of $H^2(G, A)$ are in 1-to-1 correspondence with equivalence classes of
central extensions of $G$ by $A$.
\end{lemma}

\begin{proof}[Proof Sketch]
While we need not go through this proof in detail (see \cite{IntroGroups} for
full details), we do need describe how the correspondence works since our aim is
to work in the group $H^2(G, A)$, using the elements corresponding to the two
given extensions.

For an extension $E$ of $G$ by $A$, choose a representative of each coset of $A$
in $E$ (i.e., each element of $G \cong E / A$), where we require $e$ to
represent $A$ itself. Encode these choices into a function $s : G \rightarrow
E$. Then we can define a function $f : G \times G \rightarrow A$ by $f(x,y)
\triangleq s(x) s(y) s(xy)^{-1}$. It is not hard to show that $f(x,y) \in A$ and
that $f$ is a factor set, i.e., $f \in Z^2(G, A)$.

This construction depends on the choice of representatives. Choosing a different
set of representatives, we could get a different factor set $g : G \times G
\rightarrow A$. However, if we do this, it will turn out $f - g$ is a
2-coboundary. Furthermore, the only other factor sets differing from $f$ by a
coboundary arise from other choices of representatives for the same extension.
Hence, $f + B^2(G, A)$ uniquely represents this extension.
\end{proof}

\section{Results}

\subsection{General Approach} \label{ssec:approach}

With this background, the basic idea for computing equivalence of central
extensions is simple. Given $E_1$ and $E_2$, two central extensions of $G$ by
$A$, we can compute the factor sets $f_1, f_2 \in Z^2(G, A)$ for these two
extensions using any set of representatives. As described in
\Lemma{Correspondence}, the factor sets correspond to the same extension iff
$f_1 - f_2 \in B^2(G, A)$. Thus, the general approach is to reduce extension
equivalence to testing membership in $B^2(G, A)$.

To make this concrete, we must specify what approach we use for representing
groups. Below, we present two algorithms, one classical and one quantum, for
implementing the outline just described. These algorithms differ in the approach
used to specify the input groups, with the quantum algorithm using the more
general approach of black box groups for $A$ and $E$. Specifically, we have the
following results.

\begin{theorem}
\TheoremName{classical-test}
There exists a (classical) Monte Carlo algorithm for testing the equivalence of
$E_1$ and $E_2$, two extensions of $G$ by $A$, when all groups specified by
multiplication tables, running in time $\tilde{O}(\abs{G}^6 \abs{A}^3)$.
\end{theorem}

\begin{theorem}
\TheoremName{quantum-test}
There exists a quantum algorithm for testing the equivalence of $E_1$ and $E_2$,
two extensions of $G$ by $A$, where $A$, $E_1$, and $E_2$ are given as black box
groups and $G$ is given by a multiplication table, running in time $O(\abs{G}^6
\log^6 \abs{A})$.
\end{theorem}

\begin{theorem}
\TheoremName{quantum-test-large}
There exists a quantum algorithm for testing the equivalence of $E_1$ and $E_2$,
two extensions of $G$ by $A$, where $G$ is \emph{abelian} and all groups are
presented as black box groups running in time $\poly \log \abs{G} \poly \log
\abs{A}$.
\end{theorem}

For simplicity, we first prove these three theorems in subsections
\ref{ssec:classical-alg}--\ref{ssec:quantum-alg-large} assuming that
$E_1$ and $E_2$ are \emph{central} extensions. We discuss how to extend these
two algorithms to non-central extensions in subsection~\ref{ssec:non-central}.

In subsection~\ref{ssec:hardness}, we show that the problem solved by the
quantum algorithms are classically hard under the assumption of the
Goldwasser--Micali cryptosystem \cite{GMCrypto} (that quadratic residuosity is
classically hard).

\begin{theorem}
\TheoremName{hardness}
There exists a randomized polynomial time reduction from quadratic residuosity
to testing equivalence of central extensions of $G$ by $A$, where $A$ is given
as a black box group and either $G$ is given as a multiplication table or $G$ is
abelian and given as a black box group.  Hence, under the assumption that there
is no efficient (classical) Monte Carlo algorithm for testing quadratic
residuosity, there is no efficient Monte Carlo algorithm for testing equivalence
of extensions of $G$ by $A$ in this model.
\end{theorem}

Finally, in subsection~\ref{ssec:count}, we use the machinery developed for
these algorithms to show that we can also efficiently count the number of
inequivalent extensions in the two models. Specifically, we have the following:

\begin{theorem}
\TheoremName{classical-count}
There exists an efficient (classical) Monte Carlo algorithm for counting the
number of equivalence classes of extensions of $G$ by $A$ when both groups are
given by multiplication tables.
\end{theorem}

\begin{theorem}
\TheoremName{quantum-count}
There exists an efficient quantum algorithm for counting the number of
equivalence classes of extensions of $G$ by $A$ when $A$ is given as a black box
group and $G$ is given by a multiplication table.
\end{theorem}

\subsection{Classical Algorithm} \label{ssec:classical-alg}

For the classical algorithm, we take the inputs $A$, $G$, and $E_1$ and $E_2$ as
multiplication tables. This is the usual setup for the group isomorphism
problem, and it is natural to consider extension equivalence in the same manner.
However, we must also require that the isomorphism $E_i / A \cong G$ be provided
explicitly so that we are not required to solve a group isomorphism problem in
order to understand the relationship between $E_i$ and $G$. This will be
specified as a table of pairs $(x, g)$, where each $x \in E_i$ appears exactly
once along with the $g \in G$ such that $x + A \xrightarrow{\sim} g$.

\begin{proof}[Proof of \Theorem{classical-test}]
As described above, we will reduce to membership testing in $B^2(G, A)$. Since
the group $B^2(G, A)$ has size $\sim \abs{A}^{\abs{G}}$, we cannot reduce to a
membership test using a multiplication table because the time to write such a
table is exponentially large in the input size. We also cannot reduce to a
membership test using a black box model simply because there is no efficient
classical algorithm known for membership testing in this model.  Fortunately, we
will see that we can reduce to a membership test using the third approach, a
permutation representation. We can then perform the membership testing
efficiently using the algorithm from \cite{FastPermGroupAlg}.

First, note that we can represent $A$ using the regular representation, that is,
each $a \in A$ is represented as a permutation $\sigma(a)$ of the set $A$
itself. The degree of this representation is $\abs{A}$, which is small. And it
is easy to see that this representation of $A$ is faithful. (This is Cayley's
theorem.)

Define $C^2(G, A)$ to be all maps $G \times G \rightarrow A$. These are simply
vectors of $\abs{G}^2$ elements of $A$. (Since $B^2(G, A) \le Z^2(G, A) \le
C^2(G, A)$, we can think of elements of $B^2(G, A)$ and $Z^2(G, A)$ in the same
way.) Put another way, $C^2(G, A)$ is a direct sum of $\abs{G}^2$ copies of $A$.
Hence, we can represent $f \in C^2(G, A)$ as the direct sum (as vector spaces)
of $\sigma(f(g, h))$ for each $g, h \in G$. It is again clear that this
representation is faithful: $\sigma(f)$ is the identity iff $\sigma(f(g, h))$ is
identity for each $g, h \in G$ iff $f(g, h) = e$ for each $g, h \in G$ (since
our representation of $A$ is faithful) iff $f$ is the identity in $C^2(G, A)$
(by definition).

In other words, our representation space is the set $\set{a_{g, h}}{a \in A,\ g,
h \in G}$ --- elements of $A$ labelled by pairs $(g, h) \in G \times G$. We can
see that the degree of this representation is $n \triangleq \abs{A} \abs{G}^2$.

It is possible that $A$ may have a permutation representation with smaller
degree in special cases, but in the worst case, it must be $\abs{A}$. In
particular, any simple cyclic group requires this degree. It is also easy to see
that any faithful representation of $C^2(G, A)$ must contain all $\abs{G}^2$
copies of this representation. Hence, our degree of $\abs{A} \abs{G}^2$ cannot
in general be improved.

In order to invoke a membership test for $B^2(G, A)$, we also need to provide a
generating set. The easiest way to do this is to take a generating set for
$C^1(G, A)$ and then push it forward to $B^2(G, A)$ by applying $\partial$. Any
$f \in B^2(G, A)$ satisfies $f = \partial s$ for some $s \in C^1(G, A)$. So if
$s_1, \dots, s_k$ is a generating set for $C^1(G, A)$, then we have $s =
s_1^{j_1} \dots s_k^{j_k}$ for some $\{j_i\} \subset \Integer_+$. And since
$\partial$ is a homomorphism, we have $f = \partial(s_1^{j_1} \dots s_k^{j_k}) =
\partial(s_1)^{j_1} \dots \partial(s_k)^{j_k}$. Thus, $\partial s_1, \dots,
\partial s_k$ is a generating set for $B^2(G, A)$.\footnote{Since $\partial$ is
not an isomorphism, this generating set may be redundant. However, since its
kernel is very small compared to $\abs{C^1(G, A)}$, this increases the size of
the generating set by a $1 - o(1)$ factor.}

It is easy to find a minimal generating set for $C^1(G, A)$. Since this group is
simply a direct sum of $\abs{G}$ copies of $A$, a minimal generating set for
$C^1(G, A)$ is given by $\abs{G}$ copies of a minimal generating set for $A$.
We can find a generating set for $A$ with high probability simply by choosing
$O(\log \abs{A})$ random elements \cite{PermGroupAlg}. And it is easy to see
that we can choose random elements from $A$ since we have an explicit list of
its elements. Hence, we can construct a generating set for $C^1(G, A)$
of size $O(\abs{G} \log \abs{A})$.

Finally, note that, since we have a simple formula for $\partial$, taking
constant time to evaluate for each $(g, h) \in G \times G$, we can construct the
generating set for $B^2(G, A)$ in $O(\abs{G}^2)$ time for each element in the
set. Since this set contains $O(\abs{G} \log \abs{A})$ elements, we can
construct the generating set in $O(\abs{G}^3 \log \abs{A})$ time.

The other input to the membership test is the element $f_1 - f_2 \in Z^2(G, A)$.
We can compute this easily in linear time once we construct a factor set $f_i$
for each extension. To do this, we simply need to choose (arbitrarily) a
representative $s_i(g) \in E_i$ for each $g \in G$, which we can do in one pass
over the table providing the isomorphism $E_i / A \cong G$. (Also note that we
must choose $e \in E$ to represent $e \in G$.) This takes $O(\abs{E}) =
O(\abs{A}\abs{G})$ time. Next, we compute $f_i$ for each $g, h \in G$ by $f_i(g,
y) = s(g) + s(h) - s(gh)$. Finally, we subtract them pointwise to compute $f_1 -
f_2$. All of the above be done in $O(\abs{A}\abs{G} + \abs{G}^2)$ time.

It remains to invoke a membership test for a permutation group. The fastest
algorithms \cite{PermGroupAlg} apply to so-called ``small-base groups'', but
unfortunately, this representation is not one.\footnote{The group $B^2(G, A)$
would be small-base if $\log \abs{B^2(G, A)} = O(\poly \log n) = O(\poly(\log
\abs{G} + \log \abs{A}))$, but we can see that $B^2(G, A)$ is much bigger than
this.} For the general case, the fastest known algorithm is from
\cite{FastPermGroupAlg} and runs in time $\tilde{O}(n^3)$.

All of the membership test algorithms for permutation groups work by first
computing what is called a strong generating set. As noted in
\cite{FastPermGroupAlg}, Gaussian elimination is a special case of this
construction, so the running time of $\tilde{O}(n^3)$ is in fact optimal for all
algorithms that work in this manner.

We note that the time to run this membership test dominates the time required to
prepare its inputs, so the overall running time will be $\tilde{O}(n^3) =
\tilde{O}(\abs{A}^3 \abs{G}^6)$.
\end{proof}

\subsection{Quantum Algorithm for Small $G$} \label{ssec:quantum-alg}

For classical algorithms, we excluded the possibility of using a membership test
for black box groups because no efficient algorithm is known to exist.  However,
in the quantum case, we have such an algorithm \cite{DecompAbGroups}. As a
result, it is natural to consider whether extension equivalence can also be
solved in the black box model.

Our quantum algorithm will take the inputs $A$ and $E$ as black box groups. That
is, we are given a generating set for each and an oracle for performing the
three operations listed earlier in the group $E$.\footnote{This also works for
$A$ since $A \le E$.}

For the group $G$, on the other hand, we first consider the case when $G$ is
given by a multiplication table. In this case, we can efficiently work with the
group $B^2(G, A)$ since it has a generating set of size $O(\abs{G}^2 \log
\abs{A})$ and we only need a running time polynomial in $\abs{G}$ in this model.
Practically speaking, this means that we will be able to compute equivalence of
extensions of a small group $G$ by a large group $A$ using this algorithm. Such
extensions can still be quite complicated groups.

Finally, the isomorphism $E_i / A \cong G$ will be provided as an oracle since
we cannot reasonably take a table with $\abs{E}$ rows as input. Given an element
$x \in E_i$, the oracle return the $g \in G$ corresponds to $x + A \in E_i / A$.

\begin{proof}[Proof of \Theorem{quantum-test}]
As in the classical algorithm, we will apply the correspondence in
\Lemma{Correspondence} and reduce to a membership test in $B^2(G, A)$.

In order to use a membership test for $B^2(G, A)$, we must show how to construct
an oracle for this group or a larger group containing it. We will work with
$C^2(G, A)$. Since each element of $C^2(G, A)$ is a vector (or direct sum) of
$\abs{G}^2$ elements of $A$, we can identify elements of this group by strings
containing $\abs{G}^2$ strings for elements of $A$. We can perform
multiplication and inverses pointwise, each using $\abs{G}^2$ calls to the
oracle for $A$. Similarly, the identity in $C^2(G, A)$ is simply $\abs{G}^2$
copies of the identity in $A$, so we can also check for the identity with
$\abs{G}^2$ calls to the oracle for $A$.

One input to the membership test is a generating set for $B^2(G, A)$. We saw in
the previous section that this can be constructed simply by making $\abs{G}^2$
labelled copies of a generating set for $A$. In this case, we are given a
generating set for $A$ as input, and we can turn this into $\abs{G}^2$ labelled
copies in $O(\abs{G}^2 \log \abs{A})$ time.\footnote{This is assuming that we
are given a generating set for $A$ of size $O(\log \abs{A})$. We can easily
reduce to a generating set of this size, if this is not what we are given, by
using random subproducts as described in \cite{PermGroupAlg}.}

The other input to the membership test is the element $f_1 - f_2$. As before, in
order to compute these factor sets, we need to be able to choose a
representative of each coset of $A$ in $E$. However, note that our classical
algorithm ran in $O(\abs{E})$ time, which is no longer efficient in this model.
So we will need a slightly different approach.

Instead of enumerating $E$, we will select random elements from $E$ and invoke
the oracle we are given to find the projection in $G$. If $x \in E$ projects
onto $g \in G$, then this gives us our representative $s(g) = x$ for $g$. We
continue to select random elements until we have a representative for each $g
\in G$ (aside from $e \in G$, which we set to $s(e) = e$).

Now, since we are only given a generating set for $E$, it is not possible to
select uniformly random elements. However, we can compute nearly uniformly
random elements as described in $\cite{RandomGroupElem}$ in time linear in the
size of the generating set for $E$ (plus an $\tilde{O}(\log^5 \abs{A})$ additive
term). The generated elements are nearly uniform in the sense that the
probability of generating $x \in E$ is off by a $1 - o(1)$ factor, which we can
choose to be arbitrarily small.

With this, the probability of producing any particular $g \in G$ will be $(1 \pm
\epsilon)/\abs{G}$. Hence, by standard calculations, we will produce a
representative for each $g \in G$ with high probability after $O(\abs{G} \log
\abs{G})$ random choices. The overall time to compute these representatives if
$\tilde{O}(\abs{G} \log \abs{A} + \log^5 \abs{A})$.

With choices of representatives $s_i$ for each $E_i$, we can compute the factor
sets $f_i$ and their difference $f_1 - f_2$ in the same manner as in the
classical algorithm. This takes time $O(\abs{G}^2)$.

To perform the membership test, we apply the algorithm from
\cite{DecompAbGroups}, which can be used to compute the size of a subgroup.
We call this once with the generating set for $B^2(G, A)$ and once with this
generating set plus $f_1 - f_2$. If the latter subgroup is larger, then $f_1 -
f_2 \notin B^2(G, A)$, and the extensions are not equivalent. Otherwise, they
are equivalent.

As described in \cite{QuantumCompAlg}, the running time of the algorithm for
computing group size depends on the size of the generating set, $k$, and the
maximum order of any element in the group, $q$. As mentioned above, we have $k =
O(\abs{G}^2 \log \abs{A})$ for the first. For the second, the best bound we have
in general is $q = \abs{A}$.

The algorithm first performs $O(k \log q)$ group operations. Each of these
translates into $\abs{G}^2$ calls to the oracle for $A$. Thus, all together, it
will perform $O(\abs{G}^4 \log^2 \abs{A})$ calls to the oracle for $A$.
The algorithm also performs $O(k^3 \log^2 q) = O(\abs{G}^6 \log^5 \abs{A})$
other elementary operations as part of its post-processing, which dominates the
running time.

There are a few other details about the running time of this algorithm that need
to be considered. However, to keep this presentation simpler, we discuss those
in the appendix, in section~\ref{sec:group-size}. Here, it suffices here to say
that the other necessary processing adds at most a $\log \abs{A}$ factor to the
running time, giving us a running time of $O(\abs{G}^6 \log^6 \abs{A})$.
\end{proof}

As in the classical case, it turns out that the quantum algorithm needs to
perform something like Gaussian elimination on a matrix.\footnote{Specifically,
computing the Smith normal form of a matrix. See \cite{DecompAbGroups} for
details.} This occurs within the post-processing steps of the algorithm for
computing the size of the subgroup.  The matrix in question has rows and columns
indexed by generators, and since we have $O(\abs{G}^2 \log \abs{A})$ generators,
we get an $O(\abs{G}^6)$ factor in the running time of the algorithm.

The dependence on $\abs{A}$, on the other hand, is exponentially improved
compared to the classical algorithm. Hence, if the group $G$ is fairly small
(i.e., $\abs{G} = O(\log \abs{A})$) then the quantum algorithm is exponentially
faster overall. As we will see in the next section, extensions of small groups
(even constant sized) are complicated and interesting objects.

\subsection{Quantum Algorithm for Large, Abelian $G$}
\label{ssec:quantum-alg-large}

As mentioned in previous subsection,
when $G$ is a black box group, we have little hope of working with the group
$B^2(G, A)$ since we cannot efficiently write down a generating set. Worse, we
cannot even write down an $f \in Z^2(G, A)$ corresponding to our extension
because this requires $\abs{G}$ numbers in the general case. Hence, it is clear
that we will need to put some restrictions on the form of $f$ if we are to work
with it efficiently. Below, we will see that this can be done without loss of
generality in the case where $G$ is abelian.

By the structure theorem for abelian groups, we know that $G \cong
\Integer_{d_1} \times \dots \times \Integer_{d_m}$ for some integers $d_1 \mid
d_2 \mid \dots \mid d_m$, which means $m = O(\log \abs{G})$. We can use the
algorithm of \cite{DecompAbGroups} to efficiently decompose $G$ into a product
of this form on a quantum computer, so we can assume that we have $G$ in this
form.

As usual, we will have $f = \partial s$ for some $s : G \rightarrow E$. In particular, for $\{x_i \in \Integer_{d_i}\}_{i \in [m]}$, we will choose $s(x_1, \dots, x_m) = s_1^{x_1} \dots s_m^{x_m}$ for some $\{s_i \in E\}$ such that $s_i$ is a representative of $e_i \triangleq (0, \dots, 0, 1, 0, \dots, 0) \in G$ (where the $1$ is in the $i$-th place).  We can check that this $s$ is a valid set of representatives for $G$. Since $\pi : E \rightarrow G$ is a homomorphism, we can see that $\pi(s(x_1, \dots, x_m)) = (\pi s_1)^{x_1} \dots (\pi s_m)^{x_m} = e_1^{x_1} \dots e_m^{x_m} = (x_1, 0, \dots, 0) \dots (0, \dots, 0, x_m) = (x_1, \dots, x_m)$.

Most importantly, it is clear that we can write down the numbers $s_1, \dots, s_m$ efficiently in terms of our generators for $A$, so this gives us an efficient way to represent $s$ and $f = \partial s$.

\newcommand{\factored}{\mathcal{F}}

Let us define $\factored(G, E)$ to be the set of functions $G \rightarrow E$ of
the above form, i.e, $s \in \factored(G, E)$ iff $s(x_1, \dots, x_m) = s_1^{x_1}
\cdots s_m^{x_m}$ for some $s_1, \dots, s_m \in E$. Note that we have $s(0,
\dots, 0) = 0$, so these functions are normalized. Since $s(x_1, \dots, x_m)$ is
always a representative of $(x_1, \dots, x_m) \in G$, as we saw in the proof of
\Lemma{Correspondence}, we then always have $\partial s \in Z^2(G, A)$, that is,
$\partial \factored(G, E) \subset Z^2(G, A)$.  Likewise, if we consider the
functions $\factored(G, A)$ (with codomain $A$ rather than $E$), we see that
these are a subset of $C^1(G, A)$ --- every $s \in \factored(G, A)$ is a
1-cochain, but not every 1-cochain is in this concise form (defined in terms of
some $s_1, \dots, s_m$) --- so we define $B^2_{\factored}(G, A) \triangleq
\partial \mathcal{F}(G, A) \subset \partial C^1(G, A) = B^2(G, A)$.  (It may be
helpful to refer back to \Figure{group-cohomology} for the definitions of $C^2$,
$B^2$, $Z^2$, etc.)

The following lemma shows that it will be sufficient to work with
$B^2_\factored(G, A)$.

\begin{lemma}
\LemmaName{factored}
Suppose that $f \in \partial \factored(G, E_1)$ and $g \in \partial \factored(G, E_2)$, then $f - g \in B^2(G, A)$ iff $f - g \in B^2_\factored(G, A)$.
\end{lemma}

\begin{proof}
Since $B^2_\factored(G, A) \subset B^2(G, A)$, the reverse direction is immediate.

For the forward direction, suppose that $f - g \in B^2(G, A)$. We know that $f =
\partial s$ for some $s \in \factored(G, E_1)$. Since $g$ differs from $f$ by a
coboundary, $E_1$ and $E_2$ are equivalent extensions. This means, in
particular, that there exists an isomorphism $\tau : E_2 \rightarrow E_1$
respecting $A$ and $G$. Now, let $u \in \factored(G, E_2)$ be such that $g =
\partial u$. Then we can see that \[ \tau(g(x,y)) = \tau(\partial u(x, y)) =
\tau(u(x)u(y)u(x+y)^{-1}) = \tau u(x) \tau u(y) (\tau u(x+y))^{-1}. \] Since
$g(x,y) \in A$ and $\tau$ restricts to identity on $A$, we see that $g(x,y) =
\tau g(x,y) = (\partial \tau u)(x, y)$. Thus, $g$ can be realized as $\partial
t$ for some $t : G \rightarrow E_1$, namely, $t = \tau u$. Futhermore, since $u$
is of the form $u(x_1, \dots, x_m) = u_1^{x_1} \dots u_m^{x_m}$, we see that
$t(x_1, \dots, x_m) = \tau u(x_1, \dots, x_m) = (\tau u_1)^{x_1} \dots (\tau
u_m)^{x_m}$, which shows that $t \in \factored(G, E_1)$ with $t_i \triangleq
\tau u_i$ the representative of $e_i$ for each $i \in [m]$.

The above shows that we can restrict our attention to considering $f - g =
\partial s - \partial t$, where $s, t \in \factored(G, E_1)$. In this case, we
can compute \[ f(x)-g(y) = s(x)s(y)s(x+y)^{-1}(t(x)t(y)t(x+y)^{-1})^{-1} =
s(x)s(y)s(x+y)^{-1}t(x+y)t(y)^{-1}t(x)^{-1}. \] Now, note that
$s(x+y)^{-1}t(x+y) \in A$ since \[ \pi(s(x+y)^{-1} t(x+y)) = (\pi s(x+y))^{-1}
\pi t(x+y) = -(x+y) + (x+y) = 0 \] in $G$. Since $A$ is central in $E$, we can
move $s(x+y)^{-1}t(x+y)$ to the end. This leaves $s(y)t(y)^{-1}$ adjacent. Since
this is in $A$ for the same reason, we can rearrange this as well. Thus, we have
$f(x)-g(y) = s(x)t(x)^{-1}s(y)t(y)^{-1}s(x+y)^{-1}t(x+y)$. This is close, but
not identical, to \[ \partial(s t^{-1})(x,y) = s(x)t(x)^{-1}s(y)t(y)^{-1}
(s(x+y)t(x+y)^{-1})^{-1}, \] the only difference being the order of the last two
factors.

We can show, however, that these two terms commute. In particular, let $x =
(x_1, \dots, x_m)$. Then we have $s(x_1, \dots, x_m) = s_1^{x_1} \dots
s_m^{x_m}$ and $t(x_1, \dots, x_m) = t_1^{x_1} \dots t_m^{x_m}$ so that
$s(x)t(x)^{-1} = s_1^{x_1} \dots s_m^{x_m} t_m^{-x_m} \dots t_1^{-x_1}$. Since
$s_m$ and $t_m$ are both representatives of $e_m \in G$, we know that
$s_m^{x_m}t_m^{-x_m} \in A$, which means we can move this term to the end.
Repeating this as above, we have $s(x)t(x)^{-1} = s_1^{x_1} t_1^{-x_1} \dots
s_m^{x_m} t_m^{-x_m}$. Now, since $s_m$ and $t_m$ are both representatives of
$e_m$, they must differ by a factor of some $a_m \in A$, so we have $t_m = s_m
a_m$, which means that $s_m^{x_m} t_m^{-x_m} = s_m^{x_m} s_m^{-x_m} a_m^{-x_m}$,
and more generally, $s(x)t(x)^{-1} = a_1^{-x_1} \dots a_m^{-x_m}$. Now, if we
compute the product in the other order, we have $t(x)^{-1}s(x) = t_m^{-x_m}
\dots t_1^{-x_1} s_1^{x_1} \dots s_m^{x_m} = t_1^{-x_1} s_1^{x_1} \dots
t_m^{-x_m} s_m^{x_m}$ by the same rearranging as before, and since $t_1^{-x_1}
s_1^{x_1} = s_1^{-x_1} a_1^{-x_1} s_1^{x_1} = a_1^{-x_1}$ (using the fact that
    $A$ is central in $E_1$), we can see that $t(x)^{-1}s(x) = a_m^{-x_m} \dots
a_1^{-x_1}$. This is equal to what we computed for $s(x)t(x)^{-1}$ since $A$ is
abelian, so we have shown that $f(x)-g(y) = \partial(st^{-1})(x,y)$.

If we let $v : G \rightarrow E_1$ be defined by $v(x) = s(x)t(x)^{-1}$, then we have shown above that $f - g = \partial v$. In particular, we showed $v(x_1, \dots, x_m) = a_1^{-x_1} \dots a_m^{-x_m}$, which means that $v \in \factored(G, A)$ with $v_i = a_i^{-1}$. Thus, we have seen that $f - g \in \partial \factored(G, A) = B^2_\factored(G, A)$.
\end{proof}

The following two lemmas tell us more about what elements in these groups look
like.

\begin{lemma}
\LemmaName{factored-coboundary}
If $h \in B^2_\factored(G, A)$, then there exist $\alpha_1, \dots, \alpha_m \in
A$ such that $h(x,y) = \prod_{i=1}^m \alpha_i^{\delta_i}$, where $\delta_i = 1$
if $x_i + y_i \ge d_i$ and 0 otherwise and $\alpha_i = a_i^{d_i}$ for some
$a_i$.
\end{lemma}

\begin{proof}
If $h$ is as above, we know that $h = \partial v$ for some $v \in
C^1_\factored(G, A)$, where $v$ is of the form $v(x_1, \dots, x_m) = a_1^{x_1}
\dots a_m^{x_m}$ for some $\{a_i \in A\}$. Since $A$ is abelian, we can see that
\[ h(x, y) = v(x_1, \dots, x_m) v(y_1, \dots, y_m) v(x_1+y_1, \dots,
x_m+y_m)^{-1} = \prod_{i=1}^m a_i^{x_i} a_i^{y_i} a_i^{-(x_i + y_i) \bmod d_i}
\] because $x_i + y_i$ in $G$ is computed mod $d_i$. If $x_i + y_i < d_i$, then
the mod has no effect, and we see that $h(x, y) = e$. On the other hand, if $x_i
+ y_i \ge d_i$, then $-(x_i + y_i) \bmod d_i = -x_i - y_i + d_i$. This means
that $a_i^{x_i} a_i^{y_i} a_i^{-(x_i + y_i) \bmod d_i} = a_i^{d_i}$, so we can
see that $h(x, y) = \prod_{i=1}^m a_i^{d_i \delta_i}$, where each $\delta_i$ is
defined as in the statement of the lemma. We get the form in the statement by
defining $\alpha_i = a_i^{d_i}$.
\end{proof}

\begin{lemma}
\LemmaName{factored-cocycle}
If $f \in Z^2_\factored(G, A)$, so that $f = \partial s$ for some $s  \in
\factored(G, E)$, then there exist $\{\alpha_i \in A\}_{1 \le i \le m}$ and
$\{\beta_{i,j} \in A\}_{1 \le i < j \le m}$ such that $f(x,y) = \prod_{1 \le i
\le m} \alpha_i^{\delta_i} \prod_{1 \le i < j \le m} \beta_{i,j}^{y_i x_j} $,
where $\delta_i$ is defined as in the previous lemma, $\alpha_i = s_i^{d_i}$,
and $b_{i,j} = [s_i, s_j^{-1}]$.
\end{lemma}

\begin{proof}
By definition, we have \[ f(x,y) = s(x)s(y)s(x+y)^{-1} = s_1^{x_1} \cdots
s_m^{x_m} s_1^{y_1} \cdots s_m^{y_m} s_m^{-(x_m + y_m) \bmod d_m} \cdots
s_1^{-(x_1 + y_1) \bmod d_1}. \] As in the previous lemma, we can rewrite this
as \[ f(x,y) = s_1^{x_1} \cdots s_m^{x_m} s_1^{y_1} \cdots s_m^{y_m} s_m^{-x_m -
y_m + d_m \delta_m} \cdots s_1^{-x_1 - y_1 + d_1 \delta_1}. \]

We can begin by using the fact that $s_i^{d_i} \in A$ for each $i$. This follows
because $\pi(s_i^{d_i}) = \pi(s(e_i)^{d_i}) = (0, \dots, d_i, \dots, 0) = 0$
since the $i$-th part of $G$ is $\Integer_{d_i}$, meaning addition is modulo
$d_i$.

Thus, we can define $\alpha_i \triangleq s_i^{d_i}$. Since $A$ is abelian, we
can pull all of these factors to the front. This puts $f$ in the form \[ f(x,y)
= \left( \prod_{i=1}^m \alpha_i^{\delta_i} \right) s_1^{x_1} \cdots s_m^{x_m}
s_1^{y_1} \cdots s_m^{y_m} s_m^{-x_m - y_m} \cdots s_1^{-x_1 - y_1}. \]

In the middle of the latter product, we have $s_{m-1}^{y_{m-1}} s_m^{y_m} s_m^{-x_m -y_m} s_{m-1}^{-x_{m-1} - y_{m-1}}$. We can cancel $s_m^{y_m}$ and $s_m^{-y_m}$, leaving us with $s_{m-1}^{y_{m-1}} s_m^{-x_m} s_{m-1}^{-x_{m-1} - y_{m-1}}$. In order to cancel the $s_{m-1}^{y_{m-1}}$, we first have to move it past the $s_m^{-x_m}$. We can do this by introducing a commutator that compensates for the order change. This allows the $s_{m-1}^{y_m}$ factor to cancel, leaving us with  $[s_{m-1}^{y_{m-1}}, s_m^{x_m}] s_{m-1}^{-x_{m-1}}$.

More generally, we can consider $[s(u), s(v)]$ for any $u, v \in G$. We can see
that \[ \pi [s(u), s(v)] = \pi(s(u) s(v) s(u)^{-1} s(v)^{-1}) = \pi s(u) \pi
s(v) \pi s(u)^{-1} \pi s(v)^{-1} = u + v - u - v = 0, \] which means that
$[s(u), s(v)] \in A$. In particular, this means that we can move commutators to
the front.

Hence, we can simplify $s_1^{x_1} \cdots s_m^{x_m} s_1^{y_1} \cdots s_m^{y_m}
s_m^{-x_m - y_m} \cdots s_1^{-x_1 - y_1}$ by introducing commutators to move
each factor of $s_j^{-x_j}$ in front of each remaining factor of $s_i^{y_i}$.
In the example above, we saw that there was no moving required for $i=m$, while
$i=m-1$ only need to move past $j=m$. In general, will need to swap each pair of
this form with $i < j$. Each such swap introduces a commutator, but since these
are all in $A$, we can immediately move them to the front and continue swapping
these factors and canceling the matching factors until nothing remains.

Finally, note that a swap of $s_i^{y_i}$ and $s_j^{-x_j}$ can be thought of as a
number of swaps between $s_i$'s and $s_j^{-1}$'s. Since each of the $y_i$ copies
of the first must move past each of the $x_j$ copies of the second, we see that
there are $y_i x_j$ swaps overall. Thus, we can write the commutator as $[s_i,
s_j^{-1}]^{y_i x_j}$, giving us the form in the statement of the lemma.
\end{proof}

The following is the main result needed for our algorithm.

\begin{lemma}
\LemmaName{compatible}
Let $f, f' \in Z^2_\factored(G, A)$. Write these in the form of the previous
lemma with $\{\alpha_i\}, \{\beta_{i,j}\}$ for $f$ and $\{\alpha_i'\}$ and
$\{\beta_{i,j}'\}$ for $f'$. Then $f - f' \in B^2_\factored(G, A)$ iff
$\beta_{i,j} = \beta_{i,j}'$ for all $1 \le i < j \le m$ and $(\alpha_i)^{-1}
\alpha_i'$ has a $d_i$-th root in $A$.
\end{lemma}

\begin{proof}
We begin with the reverse direction. Let $a_i \in A$ be a $d_i$-th root of
$(\alpha_i)^{-1} \alpha_i'$. Recall that $\alpha_i = s_i^{d_i}$. Replacing
$s_i$ with $s_i a_i$ gives another valid set of representatives and, hence,
an extension equivalent to $f'$. Defining $f''$ using this set of
representatives gives an $\alpha_i'' = s_i^{d_i} a_i^{d_i} = \alpha_i
(\alpha_i)^{-1} \alpha_i' = \alpha_i'$. Since $f$ and $f'$ agree on the
$\beta_{i,j}$'s and including extra factors from $A$ does not change the
$\beta_{i,j}$'s (since $A$ is central and $\beta_{i,j}$ is a commutator), we see
that $f''$ and $f'$ agree on both the $\alpha_i$'s and $\beta_{i,j}$'s, so $f''
= f'$. Next, since $f$ and $f''$ arise by choosing different representatives for
the same extension, we know that $f - f'' \in B^2(G, A)$.  However, since $f,
f'' \in Z^2_\factored(G, A)$, we have $f - f'' \in B^2_\factored(G, A)$ by
\Lemma{factored}. Thus, we can see that $f - f' = (f - f'') + (f'' - f') = f -
f'' \in B^2_\factored(G, A)$.

For the forward direction, we will separately prove the two implications, that
$f - f' \in B^2_\factored(G, A)$ implies the condition on the $\beta_{i,j}$'s
and that it implies the condition on the $\alpha_i$'s.

For the condition on the $\beta_{i,j}$'s, we will prove the contrapositive.
First, suppose that $\beta_{i,j} \not= \beta_{i,j}'$ for some $i < j$. From the
formula in \Lemma{factored-coboundary}, we can see that $h(e_i, e_j) = 0$ for
any $h \in B^2_\factored(G, A)$. On the other hand, from the formula in
\Lemma{factored-cocycle}, we see that $f(e_i, e_j) = \beta_{i,j} \not=
\beta_{i,j}' = f'(e_i, e_j)$. Since every coboundary is 0 on this pair, we
conclude that $f - f' \not\in B^2_\factored(G, A)$.

Now, we prove the condition on the $\alpha_i$'s.  Suppose that $h \triangleq f'
- f \in B^2_\factored(G, A)$. From the formula in \Lemma{factored-coboundary},
writing the constants for $h$ as $\alpha_i''$, we can see that $h(e_i,
(d_i-1)e_i) = \alpha_i'' = a_i^{d_i}$. From the formula in
\Lemma{factored-cocycle}, we see that $f(e_i, (d_i-1)e_i) = \alpha_i$ and
$f'(e_i, (d_i-1)e_i) = \alpha_i'$.  Taking $f' - f = h$ at the pair $(e_i,
(d_i-1)e_i)$ and writing with multiplicative notation, we see that $\alpha_i'
(\alpha_i)^{-1} = \alpha_i'' = a_i^{d_i}$. Since $(\alpha_i)^{-1} \alpha_i' =
\alpha_i' (\alpha_i)^{-1}$ (both are in $A$), we see that the $d_i$-th root
exists.

Thus, we have seen that, if the condition on the $\beta_{i,j}$'s and
$\alpha_i$'s does not hold (so either the $\beta_{i,j}$ condition does not hold
or the $\alpha_i$ condition does not hold), it is impossible to have $f - f' \in
B^2_\factored(G, A)$.
\end{proof}

We now have the necessary tools required to prove the theorem in this case.

\begin{proof}[Proof of \Theorem{quantum-test-large}]
Assuming that we can compute a factor set in $Z^2_\factored(G, A)$ for each
extension, we only need to compute the $\alpha_i$'s and $\beta_{i,j}$'s from
\Lemma{compatible} for each factor set and check whether they satisfy the
conditions of the last lemma.

We saw in the proof of the lemma that these constants can be found simply by
evaluating the factor set at particular points. There are only $O(m^2) =
O(\log^2 \abs{G})$ constants to compute. Given the simple form of each $f \in
Z^2_\factored(G, A)$, it is clear that we can perform these evaluations
efficiently. Thus, we can efficiently determine the $\alpha_i$'s and
$\beta_{i,j}$'s.

For the $\beta_{i,j}$'s, the conditions of \Lemma{compatible} require us simply
to check equality, which we can do for each $(i,j)$ with one call to the oracle
for $A$. For the $\alpha_i$'s, on the other hand, we need to determine whether
the quotient of two $\alpha_i$'s is a $d_i$-th root.

Recalling that $A$ is an abelian group, we can switch back to additive notation. Our goal is to determine whether there exists an $a \in A$ such that $d_i a = \alpha_i' - \alpha_i$. Since $A$ is isomorphic to a product $\Integer_{n_1} \times \dots \times \Integer_{n_k}$, this splits into $k$ independent equations. For each $1 \le j \le k$, we want to find an $a_j$ such that $d_i a_j = (\alpha_i' - \alpha_i)_j \pmod{n_k}$ or, equivalently, if there exist $a_j$ and $b_j$ such that $a_j d_i + b_j n_k = (\alpha_i' - \alpha_i)_j$. Let $d$ be the greatest common denominator of $d_i$ and $n_k$. We can solve this equation iff $d$ divides $(\alpha_i' - \alpha_i)_j$.

Thus, for the $\alpha_i$'s, the conditions of \Lemma{compatible} require us to
compute the $\alpha_i$'s, split them into the parts of the direct product, and
then check whether the difference in each component is divisible by the greatest
common denominator of $d_i$ and $n_k$. We get $d_i$ by decomposing $G$ into a
direct product of cyclic groups using the algorithm of \cite{DecompAbGroups}. We
apply the same algorithm to $A$ to find $n_k$ and the $(\cdot)_j$ components
of $a_i' - \alpha_i$ needed above.\footnote{The algorithm of
\cite{DecompAbGroups} computes not only generators for the factors of the direct
product but also formulas (the vectors $\mathbf{y}_i$) for converting from the
original generators to the new ones. The map taking $\mathbf{e}_i \mapsto
\mathbf{y}_i$ is invertible, so we can efficiently compute the reverse direction
(from new generators to the original ones) as well.} simply need to check
divisibility for $O(\log \abs{G})$-bit numbers, which we can do efficiently on a
classical computer. Since the quantum algorithm of \cite{DecompAbGroups} is
efficient, we have seen that there is an efficient quantum algorithm for testing
whether the difference of two factor sets is a coboundary.

It remains to describe how to compute each factor set or, more specifically, the
representatives $s_1, \dots, s_m$ for each of the direct factors (since we can
efficiently evaluate a factor set given these numbers). As in our earlier
quantum algorithm, we can produce nearly uniformly random elements from $E$ and
then apply the oracle to find the corresponding elements of $G$. This process
gives us nearly uniformly random elements of $G$. As we have seen before, we
need only $O(\log \abs{G})$ random elements to get a set that generates all of
$G$. The key fact is that we have not only a generating set for $G$ but rather a
generating set for $G$ with each generator coming from an element in $E$.

Since these generate $G$, we know that, for each $i \in [m]$, there exists a
product that gives $e_i \in G$. The corresponding product of elements of $E$ is
thus a representative of $e_i$. To find this product, we apply the algorithm of
\cite{DecompAbGroups} to express $G$ as a direct product of cyclic groups and
get the relations for converting from the generators we have to the standard
generators for the direct factors. These relations come in the form of an
$O(\log \abs{G}) \times O(\log \abs{G})$ matrix. For each $i \in [m]$, one
column of this matrix gives the relation for generating $e_i$ as a product of
powers of $O(\log \abs{G})$ of our random elements. Since we can compute powers
efficiently and this matrix is small, we can efficiently compute this product to
get $e_i$. More importantly, we can compute the product of the elements of $E$
corresponding to these generators to produce a representative of $e_i$. This is
a valid choice for $s_i$.

In summary, we find a set of representatives $\{s_i\}$ for each extension that
allows us to efficiently compute a factor set in $Z^2_\factored(G, A)$.  Then,
we can check whether their difference lies in $B^2_\factored(G, A)$ by computing
the $\alpha_i$'s and $\beta_{i,j}$'s for each extension and checking the
conditions of the lemma. As we saw above, both of these steps can be performed
efficiently on a quantum computer.
\end{proof}

\subsection{Algorithms for Non-Central Extensions} \label{ssec:non-central}

It is not hard to extend our algorithms to general extensions, i.e., without the
assumption that $A$ is central in $E_1$ and $E_2$.

The core fact needed by both algorithms is the correspondence between
equivalence classes of extensions and elements of $H^2(G, A)$ given in
\Lemma{Correspondence}. This relationship indeed holds for general extensions
(i.e., under the assumption that $A$ is abelian but not necessarily central).
However, in the general setting, the definition of $H^2(G, A)$ is more complex.

If $E$ is an extension of $G$ by $A$ and $t \in E$ is a representative of $g \in
G$, then it does not hold that $t^{-1} a t = a$ for all $a \in A$ if $A$ is
not central. It is easy to check that $t^{-1} a t \in A$, however, and that any
two representatives of $g \in G$ define the same action $a \mapsto a^t
\triangleq t^{-1} a t$. In fact, this defines a homomorphism $\varphi : G
\rightarrow \Aut A$, as occurs in a semi-direct product.

In the general case, extensions are identified not only by the groups $G$ and
$A$ but also by $\varphi : G \rightarrow \Aut A$. Two extensions of $G$ by $A$
with action $\varphi$ are equivalent if there exists a structure preserving
isomorphism, as before. \Lemma{Correspondence} then holds using a definition of
$H^2(G, A)$ that changes the formula for $\partial$ to include $\varphi$.

In our algorithms, the only change is that we must use the new formula when
constructing a generating set for $B^2(G, A)$. This new formula is 
$(\partial f)(x, y) \triangleq f(x)^y + f(y) - f(xy)$, where the action $a^y$ of
$G$ on $A$ is given by $\varphi$. Since this action is just conjugation by a
representative and we have a representative for each $y \in G$, it is clear that
we can compute this formula just as well.  Hence, we can efficiently test
equivalence of non-central group extensions of $G$ by $A$, in both models, with
the same running times.

\subsection{Impossibility for Classical Algorithms in the Black Box Model}
\label{ssec:hardness}

In this subsection, we show that the problem solved by our quantum algorithm is
classically hard under the assumption of the Goldwasser--Micali cryptosystem
that quadratic residuosity is classically hard. Our proof is a reduction from
quadratic residuosity to testing equivalence of \emph{central} extensions.
Hence, this argues that the problem for black box groups is hard even for the
simpler case of central extensions.

\begin{proof}[Proof of \Theorem{hardness}]
The inputs to quadratic residuosity are a large number $N$ and a $y \in
\Integer_N^*$, the group of multiplicative units modulo $N$. (We are also
assured that the Jacobi symbol of $y$ is $+1$, though that will play no part in
the construction.) Both of these inputs are encoded in $O(\log N)$ bits, so an
algorithm is only efficient if it runs in $O(\poly \log N)$ time.

The objective for this problem is to determine whether $y$ has a square root in
$\Integer_N^*$, that is, whether there exists an $x \in \Integer_N^*$ such that
$y = x^2 \pmod N$. If such an $x$ exists, $y$ is called a ``quadratic residue''.
Our reduction will construct two central extensions of $\Integer_2$ by
$\Integer_N^*$ that are equivalent iff $y$ is a quadratic residue. Since
$\Integer_2$ is both small and abelian, this is a special case of \emph{both
models} we considered for quantum algorithms. Hence, this one reduction will
show that both problems are as hard as quadratic residuosity.

As mentioned above, we can create a group extension from any factor set $f :
\Integer_2 \times \Integer_2 \rightarrow \Integer_N^*$. If we know the values of
this function, then we can perform multiplication by $(x,a)(y,b) = (xy f(a,b),
a+b)$.\footnote{Note that the group operation in $\Integer_N^*$, while abelian,
is usually written as multiplication, while that of $\Integer_2$ is written as
addition. We will follow those conventions in this section. Note, however, that
we used the opposite conventions for $A$ and $G$ in earlier sections.}
It is well-known that we can perform group operations in $\Integer_N^*$ in
$O(\poly \log N)$ time, and group operations in $\Integer_2$ take constant time,
so this computation can be performed efficiently. Likewise, the inverse of $(x,
a)$, given by $(x^{-1} f(a, -a)^{-1}, -a)$, can also be computed efficiently.
Finally, we can easily check for the identity element, which is $(1, 0)$. This
shows that we can efficiently provide an oracle for these extensions, once we
have chosen their factor sets.

Each factor set provides only four outputs since $\abs{\Integer_2 \times
\Integer_2} = 4$. Furthermore, as noted in the definition, any factor set must
also satisfy $f(a,e) = f(e,b) = e$ for all $a, b \in G$. In this case, that
means that $f(0, 0) = f(0, 1) = f(1, 0) = 1$. Thus, each factor set is defined
by the single value $f(1, 1)$. We will choose one extension to have $f(1, 1) =
1$ and the other to have $f(1, 1) = y$. Since $y$ is provided in the input, it
is clear that we can efficiently compute the value $f(a, b)$ for either of these
extensions.

We should also note that, for an $f$ so defined to be a 2-cocycle, it must
satisfy the additional (odd-looking) condition provided in the definition. This
condition ranges over three variables $a, b, c \in G$, and since $\abs{G} = 2$
in this case, this provides 8 equations that must be satisfied. It is a simple
matter to write these out for the two factor sets described above and verify
that these always hold, regardless of the value of $f(1, 1)$, so we have the
freedom to choose $f(1, 1) = y$ as above.

In addition to the oracle just described, our extension equivalence test
requires descriptions of the groups $A$, $G$, and $E$. For $G = \Integer_2$, we
can be compute a multiplication table in constant time (for the first quantum
model) or we can easily construct an oracle that computes group operations in
$\Integer_2$ in constant time (for the second quantum model). For $A =
\Integer_N^*$, we can produce a generating set (with high probability) by
choosing $O(\log N)$ random elements. To do this, we simply choose random
elements of $\Integer_N$ and then check that they are in $\Integer_N^*$ by
computing the GCD with $N$. It is well-known that this can be done efficiently,
and since there is only a $o(1)$ chance that this test fails, we can produce a
generating set in $O(\poly \log N)$ time. Finally, for the group $E$, we can
again choose $O(\log N)$ random elements (since $\abs{E} = 2\abs{A}$), and since
$E$ as a set is simply $\Integer_N^* \times \Integer_2$, we can choose a
uniformly random element of $E$ by choosing $x \in \Integer_N^*$ and $a \in
\Integer_2$ uniformly, then forming $(x, a)$.

The last input we must provide for extension equivalence is the isomorphism $E_i
/ \Integer_N^* \cong \Integer_2$. This is simply the function that maps $(x, a)
\mapsto a$.  Obviously, this can be performed efficiently.

Let $E_1$ be the extension with factor set $f_1$ having $f_1(1,1) = y$ and $E_2$
be the extension with $f_2$ having $f_2(1, 1) = 1$. Then we can see that $f_1
f_2^{-1} = f_1$. Thus, these extensions are equivalent iff there exists a
cochain $s : \Integer_2 \rightarrow \Integer_N^*$ such that $\partial s = f$. By
construction, any $s$ will ensure that $\partial s(0, 0) = \partial s(0, 1) =
\partial s(1, 0) = 1$ (otherwise, they would not be valid factor sets), so we
only need $\partial s(1, 1) = f_1(1, 1) = y$. Let $x = s(1)$.\footnote{
Any (normalized) 1-cochain $s$ must have $s(0) = 1$, so 1-cochains in this case
are in 1-to-1 correspondence with the element of $\Integer_N^*$ by the mapping
$s \mapsto s(1)$.} Then $\partial s(1, 1) = s(1) s(1) s(1 + 1)^{-1} = x \cdot x
\cdot 1^{-1} = x^2$. Thus, we can see that the extensions are equivalent iff
there exists an $x \in \Integer_N^*$ such that $x^2 = y$, i.e., iff $y$ is a
quadratic residue.
\end{proof}

Note that this example shows that extending even a constant-sized group (in this
case, $\abs{G} = 2$) by a large group can introduce substantial difficulty.

\subsection{Counting Equivalence Classes of Extensions} \label{ssec:count}

In this section, we show that it is possible to compute $\abs{H^2(G, A)}$, the
number of inequivalent extensions of $G$ by $A$, using the machinery developed
earlier for testing equivalence. The size $\abs{H^2(G, A)}$ is another quantity
that is sometimes computed by hand for extensions of small groups and would be
interesting to compute for larger groups.

We start first with the quantum algorithm, which takes $A$ as a black box group
and $G$ given by a multiplication table.

\begin{proof}[Proof of \Theorem{quantum-count}]
Since $H^2(G, A) \cong Z^2(G, A) / B^2(G, A)$, we can compute the size of the
former group from the sizes of the latter two. In fact, we computed $\abs{B^2(G,
A)}$ as part of our quantum algorithm for testing equivalence, so we know how
this can be done.

To compute $\abs{Z^2(G, A)}$, we use the fact that $Z^2(G, A)
= \Kernel \partial^2$, where $\partial^2 : C^2(G, A) \rightarrow B^3(G, A)$ is
similar to the map $\partial$ ($= \partial^1$) we used above. This map is a
surjection, so the first isomorphism theorem tells us that $B^3(G, A) \cong
C^2(G, A) / Z^2(G, A)$, which means that $\abs{Z^2(G, A)} = \abs{C^2(G, A)} /
\abs{B^3(G, A)}$. From the definition, we have $\abs{C^2(G, A)} =
\abs{A}^{\abs{G}^2}$.

To compute $\abs{B^3(G, A)}$, we can use the same approach as for $B^2(G, A)$:
we take a generating set for $C^2(G, A)$, which is simply $\abs{G}^2$ copies of
the generating set for $A$ and has size $O(\abs{G}^2 \log \abs{A})$; push this
forward into $B^3(G, A)$ by applying the map $\partial^2$, which has a simple
formula; and then invoke the algorithm for computing the size of an abelian
black box group. With $\abs{B^3(G, A)}$ in hand, we can compute $\abs{Z^2(G,
A)}$ and then $\abs{H^2(G, A)}$ by arithmetic. All of these steps can be done in
$O(\poly \abs{G} \poly \log \abs{A})$ time, so this gives an efficient
algorithm.
\end{proof}

Finally, we have a classical algorithm when $A$ and $G$ are given by
multiplication tables.

\begin{proof}[Proof of \Theorem{classical-count}]
We repeat the same approach as just described for the quantum algorithm of
computing $\abs{B^2(G, A)}$ and $\abs{B^3(G, A)}$. Now, our classical algorithm
for testing equivalence did not compute $\abs{B^2(G, A)}$ as part of its
operation. However, we did show how to efficiently construct a permutation
representation for $B^2(G, A)$, and it is well-known that we can compute the
size of a permutation group efficiently \cite{PermGroupAlg}, so we can compute
the size of this group classically as well.

We can also efficiently construct a generating set for $B^3(G, A)$, just as we
did above, by taking a generating set for $C^3(G, A)$ (in the same manner as we
did for $C^2(G, A)$ in the classical case) and pushing it forward using
$\partial^2$. We can compute the size of this group efficiently as well, using
the algorithm mentioned above, and then perform the same arithmetic as above.
\end{proof}

\section{Conclusion}

In this paper, we considered the problem of testing whether two extensions of a
group $G$ by an abelian group $A$ are the same or ``equivalent.'' If both
$\abs{A}$ and $\abs{G}$ are small, then we showed that there exists an efficient
(classical) Monte Carlo algorithm for testing equivalence. On the other hand, if
$\abs{A}$ is so large that $A$ can only be provided as a black box and either
$\abs{G}$ is small or $\abs{G}$ is large and abelian, then there is still
an efficient quantum algorithm for testing equivalence, whereas no efficient
classical algorithm exists, under the assumption that there is no efficient
classical algorithm for testing quadratic residuosity.

As mentioned in the introduction, one of the motivations for studying this
problem is its relationship to the group isomorphism problem, an important open
problem in computer science. Hence, it is worth considering what light these
results shed on the group isomorphism problem.

While the isomorphism problem applies to arbitrary groups, it is widely believed
that the case of 2-nilpotent groups contains the essential hard cases.  Any such
groups are central extensions, and hence, we can apply our classical algorithm
above to test their equivalence. If the two extensions are equivalent, then they
are isomorphic. However, the opposite does not hold.

We can conclude from this that, if it is the case that testing isomorphism of
2-nilpotent groups is hard, then the hardness must come from extensions that are
isomorphic but inequivalent. Hence, it behooves us to understand further the
computational complexity of distinguishing such extensions.

\subparagraph*{Acknowledgements}

The author would like to thank Aram Harrow for many useful discussions, much
encouragement, and careful feedback on earlier drafts of this paper. Funding was
from NSF grants CCF-0916400 and CCF-1111382.

\appendix

\section{Quantum Algorithm for Computing Group Size}
\label{sec:group-size}

The quantum algorithm in subsection~\ref{ssec:quantum-alg} requires a subroutine
that computes the size of a black box group. Earlier, we cited the algorithm and
analysis of \cite{DecompAbGroups,QuantumCompAlg} but skipped some of the finer
details of how the theorems from those papers translate into a running time for
this subroutine in our algorithm. In this section, we fill in those missing
details.

The algorithm of \cite{DecompAbGroups} is not explicitly for computing the size
of the group. Rather, it is for decomposing the group into a direct product of
cyclic groups. That is, it produces a set of generators, one for each of the
direct factors. However, it is easy to compute the size of the group from this
information.

In particular, the size of the group is simply the product of the sizes of the
direct factors, and since each of these is a cyclic group, the size of each
direct factor is simply the order of the generator. Hence, we can get the size
of the group from the output of this algorithm by invoking an order finding
subroutine.

Finding order is a special case of the algorithm for computing the period of a
function, which is also described and analyzed in \cite{QuantumCompAlg}. In our
case, the function whose period we want to find is the map $n \mapsto g^n$,
where $g \in A$ is the generator whose order we are computing.  Since the order
of $g$ is bounded by $\abs{A}$, the method of repeated squaring allows us to
compute this map with $O(\log \abs{A})$ calls to the oracle for $A$.

The quantum period finding algorithm makes only one call to the function just
described, taking $O(\log \abs{A})$ time. However, it must also perform
$O(\log^2 \abs{A})$ post-processing, which dominates the running time.

To compute the size of our group, we need to find the order of all $O(\abs{G}^2
\log \abs{A})$ generators, which we can see takes $O(\abs{G}^2 \log^3 \abs{A})$
time. This is adds only a lower order term to the overall running time.

That completes the discussion of our own post-processing to compute the size of
the group. However, we will also need to perform some pre-processing.

The algorithm described in \cite{QuantumCompAlg} requires that all of the given
generators have order that is $p^k$ for some fixed prime $p$. This is done in
order to reduce the amount of quantum computation that is needed because
separation into different $p$-groups can be done classically, as we will now
describe.

We start by finding the order of each generator. As noted above, this takes
$O(\abs{G}^2 \log^3 \abs{A})$ time. Next, we factor the order using Shor's
algorithm \cite{ShorFactoring}, which takes $O(\log^3 \abs{A})$ time.
Now, suppose that the order of $g$ is $r = p_1^{j_1} \dots p_k^{j_k}$. Then, if
we let $q_\ell = \prod_{i \not= \ell} p_i^{j_i}$, then we can see that the order
of $g^{q_\ell}$ is $p_\ell^{j_\ell}$. Furthermore, we know from the Chinese
remainder theorem that any $x \in \Integer_r$ is uniquely determined by the
values $x \bmod p_\ell^{j_\ell}$ for each $\ell$. Hence, any power of $g$ can be
written uniquely as a product of powers of $g^{q_1}, \dots, g^{q_k}$.

We now have a generating set for which we know the prime power order of
each element. Thus, we can separately pass the generators for each $p$-subgroup
(those whose order is a power of $p$) to the algorithm from
\cite{DecompAbGroups}. The structure theorem for finite abelian groups tells us
that our group is a direct product of the $p$-subgroups, so we can simply
multiply their sizes to get the size of the whole group.

We can see that this pre-processing adds only a lower order term to the running
time of the algorithm. While our generating set for the whole group may have
grown, each generator adds at most a single generator to the set for each
$p$-subgroup, so the running time of the group decomposition algorithm that we
analyzed before is unchanged. The one difference is that we may need to invoke
that algorithm as many as $\log \abs{A}$ times, so this adds a factor of $\log
\abs{A}$ to our bound on the running time.

Finally, we should note that the decomposition algorithm described in
\cite{QuantumCompAlg} also mentions $O(k^2 \log q)$ classical group
multiplications (meaning multiplication in the group $\Integer_{\abs{A}}$).
This is dominated by the $O(k^3 \log q)$ part of the post-processing, which
works in the same group, so it does not add to the overall running time.

\bibliographystyle{hplain}
\bibliography{refs}

\end{document}